\documentclass[journal]{IEEEtran}
\usepackage{amsmath,amsfonts}
\usepackage{algorithmic}
\usepackage[ruled,linesnumbered]{algorithm2e}
\usepackage{array}
\usepackage[caption=false,font=normalsize,labelfont=sf,textfont=sf]{subfig}
\usepackage{textcomp}
\usepackage{stfloats}
\usepackage{url}
\usepackage{balance}
\usepackage{verbatim}
\usepackage{graphicx}
\usepackage[colorlinks]{hyperref}
\hypersetup{citecolor=blue}
\usepackage{threeparttable}
\usepackage{multirow}
\usepackage{amssymb}
\usepackage{bm} 
\usepackage{color}
\usepackage{mathtools}

\newcommand{\trans}{{\mathsf{T}}}
\newcommand{\herm}{{\mathsf{H}}}
\newcommand{\E}{\mathbb{E}}
\newcommand{\C}{\mathbb{C}}
\newcommand{\R}{\mathbb{R}}
\DeclareMathOperator{\diag}{diag}
\DeclareMathOperator{\tr}{tr}

\DeclareMathOperator{\SNR}{SNR}
\DeclareMathOperator{\Real}{Re}

\newtheorem{proposition}{Proposition}
\newtheorem{remark}{Remark}

\begin{document}

\title{Finite-Aperture Planar Fluid Antenna Array}

\author{Zhentian Zhang, 
Jingyuan Xu, 
Kai-Kit Wong,~\IEEEmembership{Fellow,~IEEE}, 
Hao Jiang,~\IEEEmembership{Senior Member,~IEEE},\\
Zaichen Zhang,~\IEEEmembership{Senior Member,~IEEE}, and
Hyundong Shin,~\IEEEmembership{Fellow,~IEEE}
\vspace{-5mm}
	
\thanks{This work of Hao Jiang is by the National Natural Science Foundation of China projects under Grant 62471238, and in part by the Postgraduate Research and Practice Innovation Program of Jiangsu Province under Grant SCJX-0529. The work of Kai-Kit Wong is supported by the Engineering and Physical Sciences Research Council (EPSRC) under Grant EP/W026813/1. The work of H. Shin is supported by the National Research Foundation of Korea (NRF) grant funded by the Korean government (MSIT) (RS-2025-00556064 and RS-2025-25442355), and by the Ministry of Science and ICT (MSIT), Korea, under the ITRC (Information Technology Research Center) support program (IITP-2025-RS-2021-II212046), supervised by the IITP (Institute for Information \& Communications Technology Planning \& Evaluation).}

\thanks{Zhentian Zhang, Jingyuan Xu and Zaichen Zhang are with the National Mobile Communications Research Laboratory, Southeast University, Nanjing, 210096, China. (e-mail: zhentianzhangzzt@gmail.com, jingyuanxu@seu.edu.cn, zczhang@seu.edu.cn).}
\thanks{Hao Jiang is with the School of Cyber Science and Engineering, Southeast University, Nanjing 210096, P. R. China (e-mail: jiang.hao@seu.edu.cn).}
\thanks{Kai-Kit Wong is with the Department of Electronic and Electrical Engineering, University College London, Torrington Place, WC1E 7JE, United Kingdom  (e-mail: kai-kit.wong@ucl.ac.uk). K. K. Wong is also affiliated with the Department of Electronic Engineering, Kyung Hee University, Yongin-si, Gyeonggi-do 17104, Republic of Korea.}
\thanks{Hyundong Shin is with the Department of Electronics and Information Convergence Engineering, Kyung Hee University, Yongin-si, Gyeonggi-do 17104, Republic of Korea (e-mail: hshin@khu.ac.kr).}	

\thanks{Corresponding authors: Kai-Kit Wong}
}

\maketitle

\begin{abstract}
Fluid antenna systems (FASs) are emerging as a reconfigurable-aperture technology that expands physical-layer design beyond fixed, rigid antenna geometries. While the \emph{fading diversity} of FASs---which exploits spatial channel fluctuations for signal enhancement and interference avoidance---has been widely studied, the \emph{geometry diversity} created by reconfigurable port placement remains far less understood, particularly for planar architectures under finite-aperture constraints. This paper develops a systematic analytical framework for finite-aperture planar fluid antenna arrays (FAAs). First, we derive a closed-form characterization of the minimum inter-port distance under uniform random placement over a rectangular aperture and show that it follows a Rayleigh law. Its mean scales as $\mathcal{O}(M^{-1})$, in sharp contrast to the $\mathcal{O}(M^{-2})$ behavior in the linear case in which $M$ represents the number of candidate ports, revealing a fundamentally more favorable packing geometry in two dimensions. Secondly, we establish a universal Cram\'{e}r-Rao bound (CRB) for joint elevation-azimuth estimation, governed by a $2\times 2$ \emph{geometric inertia matrix} whose determinant and eigenstructure fully capture the role of port placement in estimation precision. We further prove that both the trace and determinant of this matrix are invariant to the azimuth look direction. Third, we uncover an intrinsic \emph{precision--ambiguity trade-off}: maximizing the geometric determinant to minimize the CRB drives ports toward the aperture boundary, but simultaneously increases sidelobe-induced spatial ambiguity. To balance these competing effects, we propose a regularized greedy port-selection algorithm with a spatial-diversity reward, enabling $\mathcal{O}(1)$ per-candidate evaluation. Numerical results validate the analysis and illustrate that the proposed design outperforms uniform-grid and random baselines across a wide range of parameters.
\end{abstract}

\begin{IEEEkeywords}
Fluid antenna system (FAS), finite-aperture fluid planar array, Cram\'{e}r-Rao bound, inter-port distance.
\end{IEEEkeywords}

\section{Introduction}\label{sec:intro}
\subsection{Diversity Dimensions}
\IEEEPARstart{W}{ireless} systems are moving from fixed, hardware-determined physical layers to reconfigurable architectures that adapt in real time to dynamic propagation conditions. Originally proposed by Wong {\em et al.}~in \cite{wong2020_fas_limits,Wong2020}, the fluid antenna system (FAS) concept exemplifies this trend by treating the antenna as a reconfigurable physical-layer resource---through fluidic \cite{shen2024_surfacewave_fas,wang2026_em_reconfig_fas}, conductive \cite{zhu2024_fas_history}, dielectric \cite{Motovilova-2020}, or different programmable embodiments \cite{Zhang-jsac2026,liu2025_meta_fluid_optics,zhang2025_pixel_reconfig,liu2025_wideband_pixel_fas,wong2026_pixel_meet_fas}---to broaden system design and network optimization with reconfigurable radiation properties \cite{new2024tutorial,FAS3,hong2025contemporary,FAS4,FAS_wu_tuo1,CRB2}.

Early research on FAS has proposed channel models \cite{khammassi2023_analytical_fas,ramirez2024_blockcorr}, studied the diversity performance under different fading conditions \cite{new2024_fas_outage,vega2024_outage_diversity,alvim2024_alpha_mu,ghadi2023_copula_cl,ghadi2024_gaussian_copula}. Recent study also investigated the use of FAS at both ends, resulting in the multiple-input multiple-output (MIMO)-FAS setup, and examined its diversity-multiplexing tradeoff \cite{new2024_mimofas_infotheory}. An effective operation of FAS usually requires the channel state information (CSI) and therefore, CSI estimation for FAS has become an important research topic \cite{xu2024_channel_est_mmwave_fas,zhang2024_successive_bayes_fas,xu2025_sbl_fas}. Furthermore, the unique spatial agility enabled by FAS has also been shown to yield substantial gains across a wide range of paradigms, including orthogonal frequency division multiplexing (OFDM) \cite{OFDM1,OFDM2}, integrated sensing and communication (ISAC) \cite{wang2024_fas_isac_drl,zhou2024_fas_isac_wcl,zou2024_isac_tradeoff,ISAC2,FBL1}, reconfigurable intelligent surfaces (RIS) \cite{RIS1,RIS2,RIS3,RIS4}, physical-layer security \cite{tang2023_secret_fas,xu2024_jamming_fas,ghadi2024_security_fas,vega2024_secrecy_outage}, activity detection \cite{MA1,MA2}, and unsourced random access \cite{FBL2}, etc. Moreover, recent efforts have studied FAS in the finite blocklength regime \cite{FBL3,FBL4}. On the other hand, FAS has also led to the fluid antenna multiple access (FAMA) concept \cite{hong2}. Schemes such as fast FAMA \cite{MA0,FAMA1}, slow FAMA \cite{FAMA2}, coded FAMA \cite{hong2025coded,hong20255gcoded}, turbo FAMA \cite{waqar2026_turbocharging_fama,FAMA7} and compact ultra-massive array (CUMA) \cite{CUMA1,CUMA_SG} were proposed for massive connectivity without precoding and interference cancellation.

At a fundamental level, FAS is endowed with two distinct diversity mechanisms: \emph{fading diversity} and \emph{geometry diversity}. Fading diversity exploits the spatially correlated channel envelope within a given aperture, featuring deep fades and strength plateaus, and is particularly beneficial for interference avoidance and signal-strength maximization. Geometry diversity, by contrast, leverages the \emph{deterministic} spatial structure of the array manifold, which is both predictable and invariant to the propagation environment. Traditional techniques for exploiting antenna array geometry, such as beamforming \cite{beam1,mumimo-1,mumimo-2}, rely on costly hardware including phase shifters and power amplifiers \cite{beam2}. FAS circumvents this limitation by enabling beam steering through the selection of different port-position subsets \cite{beam3}. A recent study systematically analyzed geometry-driven performance gains in a linear fluid antenna array (FAA) \cite{Zhang2025}, establishing trade-offs between estimation precision and sidelobe ambiguity in one dimension. However, a dedicated investigation of geometry diversity for \emph{planar} fluid arrays, where the richer two-dimensional aperture geometry introduces qualitatively different analytical challenges, remains conspicuously absent, which motivates the present work.

\subsection{Finite-Aperture Design}
While the implementation for FAS is diverse, FAS always has a {\em finite aperture}. In fact, any antenna system should have a finite aperture which is limited by the permissible physical space. With that said, there have been antenna array designs that are not tied to any aperture budget. For example, a variety of sparse array geometries, including nested arrays, coprime arrays, and minimum-redundancy arrays \cite{Liu2017,Moffet1968}, are well known to considerably outperform the half-wavelength uniform linear array (ULA) in estimation performance. Under a \emph{finite} aperture, however, these classical design methodologies are not necessarily effective: the limited physical extent constrains the achievable geometric spread, making the interplay between port placement, estimation accuracy, and spectral ambiguity much more intricate. A theoretical framework that characterizes the geometry diversity in reconfigurable port positions in FAS and enables systematic balancing of conflicting performance objectives is thus highly desirable. \emph{Crucially, unlike the stochastic gains afforded by fading diversity, geometry diversity is not only fully predictable but also virtually guaranteed under any reasonable port configuration.}

\subsection{Contributions}
This work establishes a comprehensive analytical and algorithmic framework for the design of finite-aperture planar fluid antenna arrays (FAAs), where $M$ active ports are freely positionable within a prescribed rectangular aperture. Extending the FAS paradigm from one dimension to two dimensions introduces several challenges: (i) the minimum inter-port distance in a planar structure obeys a qualitatively different scaling law; (ii) the Cram\'{e}r-Rao bound (CRB) for joint angle estimation is governed by a \emph{matrix} of geometric statistics rather than a single scalar variance; and (iii) the feasible placement region is subject to pairwise distance constraints that preclude the simple sorting-based projection available in the linear case. Our contributions are summarized as follows:
\begin{itemize}
\item[\textbullet] \textbf{Rayleigh law for planar minimum spacing.} We derive the exact distributional characterization of the minimum inter-port distance $R_{\min}$ when $M$ ports are placed uniformly {\em at random} over a rectangular aperture of area $\mathcal{A}$. By using a geometric collision-probability argument with the Chen-Stein Poisson approximation, we prove that $R_{\min}$ follows a Rayleigh distribution with complementary cumulative density function (CCDF), probability density function (PDF), mean, and variance, all expressed in closed form. A major finding is that $\mathbb{E}[R_{\min}] \propto M^{-1}$, which is \emph{fundamentally distinct} from the $\mathcal{O}(M^{-2})$ scaling law of the linear-array counterpart. The origin of this gap is traced to the quadratic growth of the planar exclusion volume ($\pi r^{2}$) versus the linear growth ($2\delta$) in one dimension. This result provides a statistically principled benchmark for calibrating the minimum-distance design parameter $d_{\min}$ in finite-aperture simulations.

\item[\textbullet] \textbf{Universal closed-form CRB via the geometric inertia matrix.} We derive exact closed-form CRB expressions for joint elevation-azimuth estimation under an arbitrary planar port configuration. The key analytical innovation is the identification of a $2 \times 2$ \emph{geometric inertia matrix} $\bm{\mathcal{L}}_{\mathrm{geo}}(\mathbf{P}, \phi)$, constructed from the centered second-order statistics of the port positions {\em projected onto and perpendicular to the azimuth direction}. We prove that both $\operatorname{tr}(\bm{\mathcal{L}}_{\mathrm{geo}})$ and $\det(\bm{\mathcal{L}}_{\mathrm{geo}})$ are \emph{azimuth-invariant}, since the projection amounts to a unitary rotation of the port coordinates. The CRB expressions cleanly factorize into trigonometric prefactors (encoding wavefront geometry) and inertia entries (encoding the spatial distribution of ports), thereby providing a universal, closed-form benchmark applicable to \emph{any} planar port configuration and generalizing the scalar geometric variance of the linear-array case to a full matrix characterization.

\item[\textbullet] \textbf{Precision-ambiguity trade-off and regularized greedy algorithm.} Furthermore, we identify and rigorously characterize a fundamental \emph{precision-ambiguity trade-off} intrinsic to finite-aperture planar placement: maximizing the geometric determinant $\det(\bm{\mathcal{L}}_{\mathrm{geo}})$ (equivalently, minimizing the CRB) drives all ports toward the aperture boundary, which simultaneously produces large co-array gaps and elevated sidelobe levels. To navigate this trade-off, we propose a computationally efficient \emph{regularized greedy} port-selection algorithm that augments the determinant-optimal (D-optimal) criterion with a minimum-separation diversity reward. The algorithm features $\mathcal{O}(1)$ per-candidate determinant evaluation via incremental running-sum updates, a self-normalizing diversity weight $\beta = \beta_{0} \det(\bm{\mathcal{L}}_{\mathrm{geo}}(\mathcal{S}_{0})) / \mathcal{A}$ that renders the tuning parameter $\beta_{0}$ dimensionless and aperture-independent, and a total complexity of $\mathcal{O}((M-4)  N_{c} M)$ that remains negligible for practical array sizes.
\end{itemize}
%

\textit{Notations:} Throughout this paper, boldface lowercase and uppercase letters denote column vectors and matrices, respectively. Moreover, $(\cdot)^\trans$, $(\cdot)^\herm$, $(\cdot)^*$ denote transpose, Hermitian transpose, and element-wise conjugate, respectively. Also, $\|\cdot\|$ is the Euclidean norm while $\mathbf{I}_M$ is an $M\times M$ identity matrix. Finally, $\otimes$ denotes Kronecker product, and $\diag(\cdot)$ constructs a diagonal matrix with the input as the diagonal entries.

\section{System Model}\label{sec:model}
In this section, the system configurations and the signal model are explained. Specifically, the location coordinates of the ports (i.e., port placement) are considered as continuous variables instead of discrete variables.

\subsection{Finite-Aperture Planar Configuration} 
Consider a rectangular finite-aperture planar FAA at a transceiver with $M$ active ports whose two-dimensional locations are to be determined by different requirements. The physical aperture, normalized by the carrier wavelength $\lambda$, spans $W_x$ along the $x$-axis and $W_y$ along the $y$-axis, enclosing an area $\mathcal{A} = W_x W_y$. The port positions are collected in
\begin{equation}\label{eq:port_positions}
\mathbf{P} = [\bm{p}_1, \dots, \bm{p}_M]^\trans \in \R^{M \times 2},
\end{equation}
where $\bm{p}_m = (x_m, y_m)^\trans$ denotes the normalized position of the $m$-th port. To have a defined aperture, the four \textit{corner} ports are pinned at the aperture vertices:
\begin{equation}\label{eq:corner_constraint}
\bm{p}_1 = \bm{0},~\bm{p}_2 = (W_x, 0)^\trans,~\bm{p}_3 = (0, W_y)^\trans,~\bm{p}_4 = (W_x, W_y)^\trans.
\end{equation}
The remaining $M - 4$ ports are freely placed within the aperture $[0, W_x] \times [0, W_y]$. To prevent physical overlap, a minimum inter-port distance (radius) $d_{\min}$ is imposed:
\begin{equation}\label{eq:dmin_2d}
\|\bm{p}_i - \bm{p}_j\| \ge d_{\min},~\forall i \neq j.
\end{equation}

\subsection{Signal Model}
Without loss of generality, we consider line-of-sight transmission in the far field. Denote the elevation angle (measured from the array broadside, i.e., the $z$-axis) by $\theta$ and the azimuth angle by $\phi$. The direction cosines are $u = \sin\theta\cos\phi$ and $v = \sin\theta\sin\phi$. The steering vector $\bm{a}(\theta,\phi,\mathbf{P}) \in \C^M$ under port placement $\mathbf{P}$ is expressed as
\begin{equation}\label{eq:steer_2d}
\bm{a}(\theta,\phi,\mathbf{P}) = \left[ e^{-j2\pi(x_1 u + y_1 v)}, \dots, e^{-j2\pi(x_M u + y_M v)} \right]^\trans.
\end{equation}
The received signal vector $\mathbf{y}(t) \in \C^M$ at the $t$-th snapshot can be given by
\begin{equation}\label{eq:y_2d}
\mathbf{y}(t) = \bm{a}(\theta,\phi,\mathbf{P})\, s(t) + \mathbf{n}(t),
\end{equation}
where $s(t)$ denotes the transmitted source signal with average power $P_s = \E[|s(t)|^2]$ and $\mathbf{n}(t) \sim \mathcal{CN}(\mathbf{0}, \sigma_n^2 \mathbf{I}_M)$ represents the additive white Gaussian noise (AWGN). The signal-to-noise ratio (SNR) is therefore defined as $\SNR \triangleq P_s / \sigma_n^2$.

\section{Minimum Radius Distribution Analysis}\label{sec:spacing}
The minimum-distance constraint in \eqref{eq:dmin_2d} is introduced as a deterministic geometric constraint to prevent FAS from having overlapping ports. Although the physical footprint of each port is not explicitly modeled, it is still useful to identify a representative spacing scale for finite-aperture planar placement. This section analyzes the spacing statistics under an unconstrained random placement strategy over a rectangular aperture and derives a statistical benchmark for the minimum inter-port distance. This benchmark does not replace the deterministic geometric role of $d_{\min}$, but provides a calibrated reference for selecting its value in the computer simulations.

\subsection{Random Placement Setup}
In our model, $M$ active ports are placed independently and uniformly at random within the rectangular aperture $[0, W_x] \times [0, W_y]$ with area $\mathcal{A} = W_x W_y$. Let
\begin{equation}\label{eq:Rmin_def}
R_{\min} \triangleq \min_{1 \le i < j \le M} \|\bm{p}_i - \bm{p}_j\|
\end{equation}
represent the minimum pairwise Euclidean distance amongst all $\binom{M}{2} = M(M-1)/2$ port pairs. In contrast to the one-dimensional (i.e., linear) case, where the minimum spacing is naturally characterized by the order statistics of a single coordinate, the two-dimensional setting involves a \textit{pairwise distance random variable} whose exact distribution requires integration over a high-dimensional product space. To obtain a tractable characterization, we employ a geometric-probabilistic approach based on the Poisson pair-count approximation.

\subsection{Pairwise Collision Probability}
Consider an arbitrary pair $(i, j)$ with $i < j$. Since both $\bm{p}_i$ and $\bm{p}_j$ are independently and uniformly distributed over the aperture, the difference vector $\bm{d}_{ij} = \bm{p}_i - \bm{p}_j$ has a known \textit{triangular-type} distribution over the rectangle $[-W_x, W_x] \times [-W_y, W_y]$. The exact probability that $\|\bm{d}_{ij}\| \le r$ is given by the normalized area of the intersection
\begin{equation}\label{eq:exact_pair_prob}
\mathbb{P}\big(\|\bm{p}_i - \bm{p}_j\| \le r\big) = \frac{1}{\mathcal{A}^2} \int_{\mathcal{A}} \text{Area}\big(\mathcal{B}(\bm{p}_i, r) \cap \mathcal{A}\big)\, d\bm{p}_i,
\end{equation}
where $\mathcal{B}(\bm{p}_i, r) = \{\bm{q} \in \mathbb{R}^2 : \|\bm{q} - \bm{p}_i\| \le r\}$ is the disk of radius $r$ centered at $\bm{p}_i$. Evaluating \eqref{eq:exact_pair_prob} in closed form requires treating boundary clipping effects. Nonetheless, for the regime of interest, namely $r \ll \min(W_x, W_y)$, the disk $\mathcal{B}(\bm{p}_i, r)$ lies entirely within $\mathcal{A}$ for the vast majority of center positions $\bm{p}_i$. More precisely, let $\mathcal{A}_{\text{int}}(r) = [r, W_x - r] \times [r, W_y - r]$ denote the interior sub-region whose points are at least distance $r$ from every boundary. Its area is given by
\begin{align}
\text{Area}\big(\mathcal{A}_{\text{int}}(r)\big) &= (W_x - 2r)(W_y - 2r)\notag\\
&= \mathcal{A} - 2r(W_x + W_y) + 4r^2.\label{eq:interior_area}
\end{align}
For any $\bm{p}_i \in \mathcal{A}_{\text{int}}(r)$, the disk is fully contained, meaning that $\text{Area}(\mathcal{B}(\bm{p}_i, r) \cap \mathcal{A}) = \pi r^2$. The boundary strip $\mathcal{A} \setminus \mathcal{A}_{\text{int}}(r)$ has an area $2r(W_x + W_y) - 4r^2 = \mathcal{O}(r)$, and contributes at most $\pi r^2$ per point. Hence,
\begin{align}\label{eq:pair_prob_approx}
\mathbb{P}\big(\|\bm{p}_i - \bm{p}_j\| \le r\big)
&= \frac{1}{\mathcal{A}} \left[ \pi r^2 \frac{\text{Area}(\mathcal{A}_{\text{int}}(r))}{\mathcal{A}} + \mathcal{O}\!\left(\frac{r^3}{\mathcal{A}}\right) \right] \nonumber \\
&= \frac{\pi r^2}{\mathcal{A}} + \mathcal{O}\!\left(\frac{r^3}{\mathcal{A}^{3/2}}\right) \approx \frac{\pi r^2}{\mathcal{A}},
\end{align}
where the higher-order boundary correction is negligible when $r / \sqrt{\mathcal{A}} \to 0$. This ratio admits a clean geometric interpretation: the collision probability equals the fractional area of a disk of radius $r$ relative to the total aperture.

\subsection{Poisson Pair-Count Approximation}
Define the pair-indicator random variables $Z_{ij} = \mathbf{1}\{\|\bm{p}_i - \bm{p}_j\| \le r\}$ for each pair $(i,j)$ with $i < j$, and let $N_{\text{close}}(r) = \sum_{i < j} Z_{ij}$ count the total number of \textit{close pairs}. The event $\{R_{\min} > r\}$ is exactly $\{N_{\text{close}}(r) = 0\}$. The expected count is
\begin{equation}\label{eq:E_Nclose}
\mu(r) \triangleq \E\big[N_{\text{close}}(r)\big] = \binom{M}{2} \frac{\pi r^2}{\mathcal{A}} = \frac{M(M-1)\pi r^2}{2\mathcal{A}}.
\end{equation}
Now, the indicators $\{Z_{ij}\}$ are \textit{not} mutually independent. For example, if $\bm{p}_i$ is close to $\bm{p}_j$ and $\bm{p}_j$ is close to $\bm{p}_k$, then $\bm{p}_i$ and $\bm{p}_k$ are somewhat more likely to be close (triangle inequality). However, for $r$ in the neighborhood of the typical minimum spacing, only a vanishing fraction of the $\binom{M}{2}$ pairs have overlapping disks, so the pairwise correlations satisfy $\text{Cov}(Z_{ij}, Z_{ik}) = \mathcal{O}(r^4/\mathcal{A}^2)$. Summing over all dependent triplets $(i,j,k)$, we then obtain
\begin{equation}\label{eq:cov_bound}
\sum_{\substack{(i,j),(i,k) \\ j \neq k}} \text{Cov}(Z_{ij}, Z_{ik}) = \binom{M}{3} \mathcal{O}\!\left(\frac{r^4}{\mathcal{A}^2}\right) = \mathcal{O}\!\left(\frac{M^3 r^4}{\mathcal{A}^2}\right).
\end{equation}
The Chen-Stein theorem \cite{Arratia1989} guarantees that $N_{\text{close}}(r)$ converges in total variation to a Poisson random variable with mean $\mu(r)$ when $\mu(r) = \mathcal{O}(1)$ and the error in \eqref{eq:cov_bound} is $o(1)$. Both conditions are met in the regime $r = \mathcal{O}(M^{-1}\sqrt{\mathcal{A}})$, which is precisely where the minimum distance concentrates. Consequently, we have
\begin{equation}\label{eq:poisson_approx}
\mathbb{P}\big(N_{\text{close}}(r) = 0\big) \approx e^{-\mu(r)}.
\end{equation}

\begin{proposition}\label{prop:2d_min_dist}
Under the Poisson pair-count approximation (asymptotically exact as $M \to \infty$ with fixed $\mathcal{A}$), the CCDF, PDF, expected value, and variance of $R_{\min}$ are given by
\begin{itemize}
\item[(i)] CCDF---For $r \ge 0$,
\begin{equation}\label{eq:ccdf_2d}
\mathbb{P}(R_{\min} > r) \approx \exp\!\left( -\frac{M(M-1)\pi r^2}{2\mathcal{A}} \right).
\end{equation}
\item[(ii)] PDF---Differentiating the CDF gives $F_{R_{\min}}(r) = 1 - \mathbb{P}(R_{\min} > r)$, which, for $r\ge 0$, yields
\begin{equation}\label{eq:pdf_Rmin}
f_{R_{\min}}(r) = \frac{M(M-1)\pi r}{\mathcal{A}} \exp\!\left( -\frac{M(M-1)\pi r^2}{2\mathcal{A}} \right),
\end{equation}
which is a Rayleigh distribution with scale parameter
\begin{equation}\label{eq:sigma_R}
\sigma_R = \sqrt{\frac{\mathcal{A}}{M(M-1)\pi}}.
\end{equation}
\item[(iii)] Expected value---The mean of $R_{\min}$ is found as
\begin{equation}\label{eq:E_Rmin}
\E[R_{\min}] = \frac{1}{2}\sqrt{\frac{2\mathcal{A}}{M(M-1)}} = \frac{1}{2}\sqrt{\frac{2 W_x W_y}{M(M-1)}}.
\end{equation}
\item[(iv)] Variance---The variance of $R_{\min}$ is expressed as
\begin{equation}\label{eq:var_Rmin}
\mathrm{Var}(R_{\min}) = \frac{(4-\pi)\mathcal{A}}{2\pi M(M-1)}.
\end{equation}
\end{itemize}
\end{proposition}

\begin{IEEEproof}
\textit{Part (i): CCDF via geometric--probabilistic volume argument.}
The event $\{R_{\min} > r\}$ requires that \textit{every} pair satisfies $\|\bm{p}_i - \bm{p}_j\| > r$, which is equivalent to $\{N_{\text{close}}(r) = 0\}$. Applying the Poisson approximation in \eqref{eq:poisson_approx} with mean $\mu(r)$ from \eqref{eq:E_Nclose} directly yields \eqref{eq:ccdf_2d}.

\textit{Part (ii): PDF by differentiation of the CDF.}
The CDF is $F_{R_{\min}}(r) = 1 - e^{-\mu(r)}$. Let $\alpha \triangleq M(M-1)\pi / (2\mathcal{A})$ so that $\mu(r) = \alpha r^2$. Differentiating with respect to $r$,
\begin{align}\label{eq:pdf_derivation}
f_{R_{\min}}(r) &= \frac{d}{dr} F_{R_{\min}}(r)= \frac{d}{dr}\!\left(1 - e^{-\alpha r^2}\right) \nonumber \\
&= -e^{-\alpha r^2} \times \frac{d}{dr}\!\left(-\alpha r^2\right)= 2\alpha r\, e^{-\alpha r^2}.
\end{align}
Substituting $\alpha = M(M-1)\pi/(2\mathcal{A})$ back into \eqref{eq:pdf_derivation} confirms \eqref{eq:pdf_Rmin}, which is recognized as the PDF of a Rayleigh random variable $R \sim \text{Rayleigh}(\sigma_R)$ whose general form is $f_R(r) = \frac{r}{\sigma_R^2}\exp\!\left(-\frac{r^2}{2\sigma_R^2}\right)$, with matching coefficients: $\frac{1}{\sigma_R^2} = \frac{M(M-1)\pi}{\mathcal{A}}$, which gives $\sigma_R = \sqrt{\mathcal{A}/\big(M(M-1)\pi\big)}$.

\textit{Part (iii): Expected value.}
For $R \sim \text{Rayleigh}(\sigma_R)$, the mean is known to be
\begin{equation}\label{eq:rayleigh_mean_general}
\E[R] = \sigma_R \sqrt{\frac{\pi}{2}}.
\end{equation}
We verify this by direct integration. Using the substitution $t = \alpha r^2$ (so that $r = \sqrt{t/\alpha}$, $dr = \frac{1}{2\sqrt{\alpha t}}\, dt$), we get
\begin{align}\label{eq:mean_integral}
\E[R_{\min}] &= \int_0^{\infty} r\, f_{R_{\min}}(r)\, dr= 2\alpha \int_0^{\infty} r^2\, e^{-\alpha r^2}\, dr \nonumber \\
&= 2\alpha \int_0^{\infty} \frac{t}{\alpha}\, e^{-t}\, \frac{dt}{2\sqrt{\alpha t}}= \frac{1}{\sqrt{\alpha}} \int_0^{\infty} t^{1/2}\, e^{-t}\, dt \nonumber \\
&= \frac{1}{\sqrt{\alpha}}\, \Gamma\!\left(\frac{3}{2}\right)= \frac{1}{\sqrt{\alpha}}\frac{\sqrt{\pi}}{2},
\end{align}
in which $\Gamma(3/2) = \frac{1}{2}\Gamma(1/2) = \frac{\sqrt{\pi}}{2}$ is a standard Gamma function identity. Substituting $\alpha = M(M-1)\pi/(2\mathcal{A})$ gives
\begin{align}\label{eq:mean_substitution}
\E[R_{\min}] &= \frac{\sqrt{\pi}}{2} \sqrt{\frac{2\mathcal{A}}{M(M-1)\pi}} = \frac{1}{2}\sqrt{\frac{2\mathcal{A}}{M(M-1)}},
\end{align}
where the $\pi$ terms cancel, confirming \eqref{eq:E_Rmin}.

\textit{Part (iv): Variance.} The second moment is evaluated similarly. With the same substitution $t = \alpha r^2$, we have
\begin{align}\label{eq:second_moment}
\E[R_{\min}^2] &= \int_0^{\infty} r^2\, f_{R_{\min}}(r)\, dr= 2\alpha \int_0^{\infty} r^3\, e^{-\alpha r^2}\, dr \nonumber \\
&= 2\alpha \int_0^{\infty} \frac{t^{3/2}}{\alpha^{3/2}}\, e^{-t}\, \frac{dt}{2\sqrt{\alpha t}}= \frac{1}{\alpha} \int_0^{\infty} t\, e^{-t}\, dt \nonumber \\
&= \frac{1}{\alpha}\, \Gamma(2) = \frac{1}{\alpha}= \frac{2\mathcal{A}}{M(M-1)\pi}.
\end{align}
The variance then follows from $\mathrm{Var}(R_{\min}) = \E[R_{\min}^2] - \big(\E[R_{\min}]\big)^2$, which gives
\begin{align}\label{eq:variance_calc}
\mathrm{Var}(R_{\min})&= \frac{2\mathcal{A}}{M(M-1)\pi} - \frac{2\mathcal{A}}{4 M(M-1)} \nonumber \\
&= \frac{2\mathcal{A}}{M(M-1)} \left( \frac{1}{\pi} - \frac{1}{4} \right)= \frac{(4 - \pi)\mathcal{A}}{2\pi M(M-1)},
\end{align}
confirming \eqref{eq:var_Rmin}. Note that $4 - \pi \approx 0.858$, consistent with the well-known Rayleigh variance formula $\mathrm{Var}(R) = \frac{4-\pi}{2}\sigma_R^2$.
\end{IEEEproof}

\begin{remark}\label{remark:scaling_comparison}
\textit{The scaling law $\E[R_{\min}] \propto M^{-1}$ in the planar case is fundamentally different from the linear-array result $\E[\Delta_{\min}] = W_{\max}/(M^2-1) \propto M^{-2}$. The distinction originates from dimensionality: in one dimension, $M$ points on an interval produce $M-1$ gaps, and the minimum gap shrinks quadratically because the ``collision volume'' (length) scales linearly with the threshold $\delta$; in two dimensions, the collision volume scales as $\pi r^2$ (quadratically in $r$) providing each pair with a larger ``exclusion budget''. Thus, the planar aperture offers a quadratically more favorable packing geometry, partially compensating the $\binom{M}{2}$-fold growth of interacting pairs.}
\end{remark}

\begin{remark}\label{remark:dmin_guideline}
\textit{Proposition~\ref{prop:2d_min_dist} provides a principled guideline for selecting $d_{\min}$ for random placement strategy. Requiring that the constraint is active with probability no more than $\epsilon$ (i.e., $\mathbb{P}(R_{\min} \le d_{\min}) \le \epsilon$) yields, from \eqref{eq:ccdf_2d},
\begin{equation}\label{eq:dmin_guideline}
d_{\min} \le \sqrt{-\frac{2\mathcal{A}\ln(1-\epsilon)}{M(M-1)\pi}}\,.
\end{equation}
With the above result, for instance, with $M = 12$, $\mathcal{A} = 16\lambda^2$, and $\epsilon = 0.05$, this gives $d_{\min} \le 0.31\lambda$.}
\end{remark}

\subsection{Comparison with the Exact Linear-Array Result}
To further illuminate the dimensional contrast, we briefly recall the exact result for the linear FAA from \cite{Zhang2025}. For $M$ ports placed uniformly at random on $[0, W_{\max}]$, the minimum spacing $\Delta_{\min}$ has the \textit{exact} CCDF
\begin{equation}\label{eq:ccdf_1d_exact}
\mathbb{P}(\Delta_{\min} > \delta) = \left(1 - \frac{(M-1)\delta}{W_{\max}}\right)^M,~0 \le \delta \le \frac{W_{\max}}{M-1},
\end{equation}
from the deterministic volume-ratio argument (simplex transformation). In one dimension, no approximation is needed because the order-statistics structure enables a bijective, volume-preserving change of variables. In two dimensions, however, the pairwise distance constraint $\|\bm{p}_i - \bm{p}_j\| > r$ induces a complex exclusion geometry without a natural ordering, making an exact volume calculation intractable. The Poisson approximation in \eqref{eq:poisson_approx} circumvents this difficulty at the cost of a boundary correction that vanishes as $r / \sqrt{\mathcal{A}} \to 0$. A side-by-side asymptotic comparison is given in Table~\ref{tab:scaling}.

\begin{table}[t!]
\centering
\caption{Scaling comparison of minimum inter-port distance.}\label{tab:scaling}
\begin{tabular}{lcc}
\hline
\textbf{Metric} & \textbf{Linear FAA (1-D) \cite{Zhang2025}} & \textbf{Planar FAA (2-D) (this work)} \\
\hline
$\E[\cdot]$ & $\dfrac{W_{\max}}{M^2-1} \propto M^{-2}$ & $\dfrac{1}{2}\sqrt{\dfrac{2\mathcal{A}}{M(M-1)}} \propto M^{-1}$ \\[8pt]
Distribution & Beta-type & Rayleigh \\[2pt]
Derivation & Exact (simplex) & Poisson approx. \\[2pt]
Boundary & Compact support & $r \in [0, \infty)$ (approx.) \\
\hline
\end{tabular}
\end{table}

\section{Universal CRB for Planar FAA}\label{sec:crb}
In this section, we derive the exact closed-form expression for the CRB under finite-aperture planar budget. By explicitly \textit{relating the CRB to the geometric inertia matrix of port positions}, we here provide a theoretical benchmark universally applicable for different planar systems and different antenna placements. In the sequel, we omit the variable term $\mathbf{P}$ in the steering vector and treat port positions as deterministic constants to derive the CRB for the angle pair $(\theta, \phi)$.

\begin{proposition}[Universal CRB for Finite-Aperture Planar FAA]\label{prop:CRB}
Consider a planar FAA with $M$ ports at $\{(x_m, y_m)\}_{m=1}^{M}$ (in units of~$\lambda$) within a finite aperture, observing a single far-field source from direction $(\theta, \phi)$ over $T$ snapshots at operating $\mathrm{SNR} = P_s / \sigma_n^2$. Then the CRBs for the elevation and azimuth angles admit the closed-form expressions
\begin{align}
\mathrm{CRB}(\theta) &= \frac{1}{8\pi^2 T \times \mathrm{SNR} \times \cos^2\!\theta}\cdot \frac{\mathcal{L}_{rr}}{\mathcal{L}_{qq}\,\mathcal{L}_{rr} - \mathcal{L}_{qr}^2},\label{eq:prop_CRB_theta} \\
\mathrm{CRB}(\phi) &= \frac{1}{8\pi^2 T \times \mathrm{SNR} \times \sin^2\!\theta}\cdot \frac{\mathcal{L}_{qq}}{\mathcal{L}_{qq}\,\mathcal{L}_{rr} - \mathcal{L}_{qr}^2},\label{eq:prop_CRB_phi}
\end{align}
where $\mathcal{L}_{qq}$, $\mathcal{L}_{rr}$, and $\mathcal{L}_{qr}$ are the entries of the $2\times 2$ \emph{geometric inertia matrix}, given by
\begin{equation}\label{eq:prop_Lgeo}
\bm{\mathcal{L}}_{\mathrm{geo}}(\mathbf{P}, \phi) =
\begin{bmatrix}
\mathcal{L}_{qq} & \mathcal{L}_{qr} \\
\mathcal{L}_{qr} & \mathcal{L}_{rr}
\end{bmatrix},
\end{equation}
computed from the port positions projected onto and perpendicular to the azimuth look direction. Specifically, for each port~$m$, define the rotated coordinates $q_m = x_m \cos\phi + y_m \sin\phi,~r_m = -x_m \sin\phi + y_m \cos\phi,$ with centroids $\bar{q} = \frac{1}{M}\sum_{m} q_m$ and $\bar{r} = \frac{1}{M}\sum_{m} r_m$. The inertia entries are then the centered second-order statistics $\mathcal{L}_{qq} = \!\sum_{m=1}^{M}(q_m - \bar{q})^2, \mathcal{L}_{rr} = \!\sum_{m=1}^{M}(r_m - \bar{r})^2, \mathcal{L}_{qr} = \!\sum_{m=1}^{M}(q_m - \bar{q})(r_m - \bar{r})$. Both $\tr(\bm{\mathcal{L}}_{\mathrm{geo}})$ and the determinant $\det(\bm{\mathcal{L}}_{\mathrm{geo}}) = \mathcal{L}_{qq}\,\mathcal{L}_{rr} - \mathcal{L}_{qr}^2$ are invariant to the azimuth angle~$\phi$, since $q_m$ or $r_m$ is a unitary rotation of~$(x_m, y_m)$. The corresponding $2\times 2$ Fisher information matrix (FIM) that yields \eqref{eq:prop_CRB_theta} and \eqref{eq:prop_CRB_phi} upon inversion is
\begin{equation}\label{eq:prop_FIM}
\mathbf{J}(\bm{\eta}) = 8\pi^2 T\mathrm{SNR} \times
\begin{bmatrix}
\cos^2\!\theta\;\mathcal{L}_{qq} & \cos\theta\sin\theta\;\mathcal{L}_{qr} \\
\cos\theta\sin\theta\;\mathcal{L}_{qr} & \sin^2\!\theta\;\mathcal{L}_{rr}
\end{bmatrix},
\end{equation}
which is non-singular whenever $\det(\bm{\mathcal{L}}_{\mathrm{geo}}) > 0$, i.e., the ports are not collinear after projection.
\end{proposition}

\begin{IEEEproof}
We structure our proof into several parts below.
\paragraph{FIM}
Recalling the received signal from \eqref{eq:y_2d}, the parameter vector to be estimated is $\bm{\eta} = [\theta, \phi]^\trans$. Under the deterministic signal model, the $2 \times 2$ FIM with respect to $\bm{\eta}$ by the Slepian-Bangs formula \cite[Section~3.3]{Stoica2005} is given as
\begin{equation}\label{eq:fim_2d}
[\mathbf{J}]_{ij} = \frac{2TP_s}{\sigma_n^2} \Real\!\left\{ \dot{\bm{a}}_i^\herm \mathbf{P}_{\bm{a}}^{\perp} \dot{\bm{a}}_j \right\},~i,j \in \{1,2\},
\end{equation}
where $\Real\{\cdot\}$ takes the real part of the argument, $\dot{\bm{a}}_i = \partial \bm{a} / \partial \eta_i$ denotes the partial derivative of the steering vector with respect to the $i$-th parameter, and $\mathbf{P}_{\bm{a}}^{\perp}$ is the orthogonal projection matrix onto the null space of $\bm{a}(\theta,\phi)$, given by
\begin{equation}\label{eq:Pa_2d}
\mathbf{P}_{\bm{a}}^{\perp} = \mathbf{I}_M - \frac{\bm{a}(\theta,\phi)\bm{a}^\herm(\theta,\phi)}{\bm{a}^\herm(\theta,\phi)\bm{a}(\theta,\phi)} = \mathbf{I}_M - \frac{1}{M}\bm{a}\bm{a}^\herm,
\end{equation}
where the simplification uses $|[\bm{a}]_m| = |e^{-j2\pi(x_m u + y_m v)}| = 1$ and hence $\bm{a}^\herm\bm{a} = M$. The FIM in \eqref{eq:fim_2d} is a $2 \times 2$ symmetric matrix with four entries $(J_{\theta\theta}, J_{\phi\phi}, J_{\theta\phi}, J_{\phi\theta})$, where $J_{\theta\phi} = J_{\phi\theta}$ by symmetry. Next, we derive each entry step-by-step below.

\paragraph{Derivative of the Steering Vector with Respect to $\theta$} To obtain an explicit relationship between estimation accuracy and array geometry, we must evaluate the derivatives in the quadratic forms of \eqref{eq:fim_2d}. Recalling the direction cosines $u = \sin\theta\cos\phi$ and $v = \sin\theta\sin\phi$, the $m$-th element of the steering vector in \eqref{eq:steer_2d} is expressed as
\begin{equation}\label{eq:am_element}
[\bm{a}]_m = e^{-j2\pi(x_m u + y_m v)} = e^{-j2\pi(x_m \sin\theta\cos\phi + y_m \sin\theta\sin\phi)}.
\end{equation}
Its partial derivative with respect to $\theta$ can be computed as
\begin{align}\label{eq:adot_theta_full}
\frac{\partial [\bm{a}]_m}{\partial \theta}&= e^{-j2\pi(x_m u + y_m v)} (-j2\pi) \times\notag\\
&\quad\quad\frac{\partial}{\partial\theta}\big(x_m \sin\theta\cos\phi + y_m \sin\theta\sin\phi\big) \nonumber \\
&= [\bm{a}]_m (-j2\pi) \cos\theta \big(x_m \cos\phi + y_m \sin\phi\big).
\end{align}
To compactly express this result, we introduce the \textit{projected positions} along the azimuth look direction:
\begin{equation}\label{eq:proj_q}
q_m \triangleq x_m \cos\phi + y_m \sin\phi,
\end{equation}
which represents the projection of the $m$-th port onto the unit vector $(\cos\phi, \sin\phi)^\trans$. With this notation, \eqref{eq:adot_theta_full} simplifies to
\begin{equation}\label{eq:adot_theta_compact}
\frac{\partial [\bm{a}]_m}{\partial \theta} = -j2\pi (\cos\theta) \times q_m \times [\bm{a}]_m.
\end{equation}
Defining $\mathbf{D}_q = \diag(q_1, q_2, \dots, q_M)$, the full derivative vector can be written as
\begin{equation}\label{eq:adot_theta}
\dot{\bm{a}}_\theta \triangleq \frac{\partial \bm{a}}{\partial \theta} = -j2\pi\cos\theta\, \mathbf{D}_q\, \bm{a}.
\end{equation}

\paragraph{Derivative of the Steering Vector with Respect to $\phi$} Similarly, the partial derivative of $[\bm{a}]_m$ with respect to $\phi$ is
\begin{align}\label{eq:adot_phi_full}
\frac{\partial [\bm{a}]_m}{\partial \phi}&= [\bm{a}]_m (-j2\pi) \frac{\partial}{\partial\phi}\big(x_m \sin\theta\cos\phi + y_m \sin\theta\sin\phi\big) \notag\\
&= [\bm{a}]_m (-j2\pi) \sin\theta \big(-x_m \sin\phi + y_m \cos\phi\big).
\end{align}
We define the projected position along the direction \textit{perpendicular} to the azimuth look direction:
\begin{equation}\label{eq:proj_r}
r_m \triangleq -x_m \sin\phi + y_m \cos\phi,
\end{equation}
which is the projection of the $m$-th port onto the unit vector $(-\sin\phi, \cos\phi)^\trans$. Note that $(q_m, r_m)$ is obtained from $(x_m, y_m)$ by a planar rotation of angle $\phi$:
\begin{equation}\label{eq:rotation}
\begin{bmatrix} q_m \\
r_m
\end{bmatrix}
= \underbrace{\begin{bmatrix} \cos\phi & \sin\phi \\ -\sin\phi & \cos\phi \end{bmatrix}}_{\triangleq\, \mathbf{R}(\phi)}
\begin{bmatrix} x_m \\ 
y_m
\end{bmatrix},
\end{equation}
where $\mathbf{R}(\phi)$ is a unitary rotation matrix with $\det(\mathbf{R}) = 1$, preserving distances and areas. Similarly, with notation $\mathbf{D}_r = \diag(r_1, r_2, \dots, r_M)$, the full derivative vector is given by
\begin{equation}\label{eq:adot_phi}
\dot{\bm{a}}_\phi \triangleq \frac{\partial \bm{a}}{\partial \phi} = -j2\pi\sin\theta\, \mathbf{D}_r\, \bm{a}.
\end{equation}

\paragraph{Expansion of the Quadratic Forms}
The FIM in \eqref{eq:fim_2d} involves three distinct quadratic forms: $\mathcal{Q}_{\theta\theta} = \dot{\bm{a}}_\theta^\herm \mathbf{P}_{\bm{a}}^{\perp} \dot{\bm{a}}_\theta$, $\mathcal{Q}_{\phi\phi} = \dot{\bm{a}}_\phi^\herm \mathbf{P}_{\bm{a}}^{\perp} \dot{\bm{a}}_\phi$, and $\mathcal{Q}_{\theta\phi} = \dot{\bm{a}}_\theta^\herm \mathbf{P}_{\bm{a}}^{\perp} \dot{\bm{a}}_\phi$. We derive the first in full detail and the remaining two follow analogously.

Substituting \eqref{eq:Pa_2d} into $\mathcal{Q}_{\theta\theta}$, we get
\begin{equation}\label{eq:Qtt_expand}
\mathcal{Q}_{\theta\theta} = \dot{\bm{a}}_\theta^\herm \left(\mathbf{I}_M - \frac{1}{M}\bm{a}\bm{a}^\herm\right) \dot{\bm{a}}_\theta= \underbrace{\dot{\bm{a}}_\theta^\herm \dot{\bm{a}}_\theta}_{\text{Term I}} - \frac{1}{M}\underbrace{\dot{\bm{a}}_\theta^\herm \bm{a}\bm{a}^\herm \dot{\bm{a}}_\theta}_{\text{Term II}}.
\end{equation}
First, we evaluate Term I. This can be done by plugging \eqref{eq:adot_theta} and recalling that $|[\bm{a}]_m| = 1$, giving
\begin{align}\label{eq:TermI_tt}
\dot{\bm{a}}_\theta^\herm \dot{\bm{a}}_\theta&= (2\pi\cos\theta)^2 \bm{a}^\herm \mathbf{D}_q^\herm \mathbf{D}_q \bm{a} \nonumber \\
&= 4\pi^2\cos^2\theta \sum_{m=1}^M q_m^2 \underbrace{|[\bm{a}]_m|^2}_{=1}= 4\pi^2\cos^2\theta \sum_{m=1}^M q_m^2.
\end{align}
For Term II, we first compute the inner product
\begin{align}\label{eq:inner_a_adot_theta}
\bm{a}^\herm \dot{\bm{a}}_\theta&= \sum_{m=1}^M [\bm{a}]_m^* \big(-j2\pi\cos\theta \cdot q_m \cdot [\bm{a}]_m\big) \nonumber \\
&= -j2\pi\cos\theta \sum_{m=1}^M \underbrace{[\bm{a}]_m^* [\bm{a}]_m}_{=1} q_m= -j2\pi\cos\theta \sum_{m=1}^M q_m.
\end{align}
As a result, Term II can be evaluated as
\begin{align}\label{eq:TermII_tt}
\dot{\bm{a}}_\theta^\herm \bm{a}\bm{a}^\herm \dot{\bm{a}}_\theta= |\bm{a}^\herm \dot{\bm{a}}_\theta|^2= 4\pi^2\cos^2\theta \left(\sum_{m=1}^M q_m\right)^2.
\end{align}
Substituting \eqref{eq:TermI_tt} and \eqref{eq:TermII_tt} back into \eqref{eq:Qtt_expand}, we then obtain
\begin{equation}\label{eq:Qtt_combined}
\mathcal{Q}_{\theta\theta} = 4\pi^2\cos^2\theta \left[\sum_{m=1}^M q_m^2 - \frac{1}{M}\left(\sum_{m=1}^M q_m\right)^2\right].
\end{equation}
Let $\bar{q} = \frac{1}{M}\sum_{m=1}^M q_m$ represent the centroid of the projected positions along the azimuth direction. Applying the variance identity $\sum_{m} (q_m - \bar{q})^2 = \sum_{m} q_m^2 - \frac{1}{M}(\sum_{m} q_m)^2$, the bracketed term in \eqref{eq:Qtt_combined} is the \textit{geometric variance along the $q$-axis}:
\begin{equation}\label{eq:Lqq_def}
\sum_{m=1}^M q_m^2 - \frac{1}{M}\left(\sum_{m=1}^M q_m\right)^2 = \sum_{m=1}^M (q_m - \bar{q})^2 \triangleq \mathcal{L}_{qq}.
\end{equation}
Consequently, $\mathcal{Q}_{\theta\theta} = 4\pi^2\cos^2\theta\, \mathcal{L}_{qq}$. Noting that $\mathcal{Q}_{\theta\theta}$ is real and non-negative, $\Real\{\mathcal{Q}_{\theta\theta}\} = \mathcal{Q}_{\theta\theta}$.

By the same procedure applied to $\dot{\bm{a}}_\phi$, one can obtain
\begin{align}
\dot{\bm{a}}_\phi^\herm \dot{\bm{a}}_\phi &= 4\pi^2\sin^2\theta \sum_{m=1}^M r_m^2, \label{eq:TermI_pp} \\
|\bm{a}^\herm \dot{\bm{a}}_\phi|^2 &= 4\pi^2\sin^2\theta \left(\sum_{m=1}^M r_m\right)^2, \label{eq:TermII_pp}
\end{align}
yielding
\begin{equation}\label{eq:Qpp_final}
\mathcal{Q}_{\phi\phi} = 4\pi^2\sin^2\theta\, \mathcal{L}_{rr},
\end{equation}
where $\mathcal{L}_{rr} \triangleq \sum_{m=1}^M (r_m - \bar{r})^2$ with $\bar{r} = \frac{1}{M}\sum_{m} r_m$.

\paragraph{Cross-Term Derivation} The off-diagonal FIM element $J_{\theta\phi}$ requires the cross quadratic form $\mathcal{Q}_{\theta\phi} = \dot{\bm{a}}_\theta^\herm \mathbf{P}_{\bm{a}}^{\perp} \dot{\bm{a}}_\phi$. To proceed, we expand it with the projector so that
\begin{equation}\label{eq:Qtp_expand}
\mathcal{Q}_{\theta\phi} = \underbrace{\dot{\bm{a}}_\theta^\herm \dot{\bm{a}}_\phi}_{\text{Term III}} - \frac{1}{M}\underbrace{(\dot{\bm{a}}_\theta^\herm \bm{a})(\bm{a}^\herm \dot{\bm{a}}_\phi)}_{\text{Term IV}}.
\end{equation}
To evaluate Term III, we substitute \eqref{eq:adot_theta} and \eqref{eq:adot_phi} to give
\begin{align}\label{eq:TermIII}
\dot{\bm{a}}_\theta^\herm \dot{\bm{a}}_\phi
&= (j2\pi\cos\theta)\, \bm{a}^\herm \mathbf{D}_q (-j2\pi\sin\theta)\, \mathbf{D}_r\, \bm{a} \nonumber \\
&= 4\pi^2\cos\theta\sin\theta\, \bm{a}^\herm \mathbf{D}_q \mathbf{D}_r\, \bm{a} \nonumber \\
&= 4\pi^2\cos\theta\sin\theta \sum_{m=1}^M q_m r_m \underbrace{|[\bm{a}]_m|^2}_{=1} \nonumber \\
&= 4\pi^2\cos\theta\sin\theta \sum_{m=1}^M q_m r_m,
\end{align}
where we have used $\mathbf{D}_q \mathbf{D}_r = \diag(q_1 r_1, \dots, q_M r_M)$.

To obtain Term IV, we start from \eqref{eq:inner_a_adot_theta} and its $\phi$-counterpart $\bm{a}^\herm\dot{\bm{a}}_\phi = -j2\pi\sin\theta\sum_{m} r_m$ such that
\begin{align}\label{eq:TermIV}
(\dot{\bm{a}}_\theta^\herm \bm{a})(\bm{a}^\herm \dot{\bm{a}}_\phi)
&= \overline{(\bm{a}^\herm\dot{\bm{a}}_\theta)} (\bm{a}^\herm\dot{\bm{a}}_\phi) \nonumber \\
&= \big(j2\pi\cos\theta \textstyle\sum_{m} q_m\big)\big(-j2\pi\sin\theta \textstyle\sum_{m} r_m\big) \nonumber \\
&= 4\pi^2\cos\theta\sin\theta \left(\textstyle\sum_{m} q_m\right)\left(\textstyle\sum_{m} r_m\right).
\end{align}
Substituting \eqref{eq:TermIII} and \eqref{eq:TermIV} into \eqref{eq:Qtp_expand} gives
\begin{multline}\label{eq:Qtp_combined}
\mathcal{Q}_{\theta\phi} = 4\pi^2\cos\theta\sin\theta\\
\times \left[\sum_{m=1}^M q_m r_m - \frac{1}{M}\left(\sum_{m=1}^M q_m\right)\!\left(\sum_{m=1}^M r_m\right)\right].
\end{multline}
Applying the cross-covariance identity $\sum_{m}(q_m - \bar{q})(r_m - \bar{r}) = \sum_{m} q_m r_m - M\bar{q}\bar{r} = \sum_{m} q_m r_m - \frac{1}{M}(\sum_{m} q_m)(\sum_{m} r_m)$, the sum inside the bracket above equals the \textit{geometric cross-inertia}, which can be obtained as
\begin{multline}\label{eq:Lqr_identity}
\sum_{m=1}^M q_m r_m - \frac{1}{M}\left(\sum_{m=1}^M q_m\right)\!\left(\sum_{m=1}^M r_m\right) \\
= \sum_{m=1}^M (q_m - \bar{q})(r_m - \bar{r}) \triangleq \mathcal{L}_{qr}.
\end{multline}
Since $\mathcal{Q}_{\theta\phi}$ involves only real-valued position statistics multiplied by a real trigonometric prefactor, we have $\Real\{\mathcal{Q}_{\theta\phi}\} = \mathcal{Q}_{\theta\phi}$, giving
\begin{equation}\label{eq:Qtp_final}
\mathcal{Q}_{\theta\phi} = 4\pi^2\cos\theta\sin\theta\, \mathcal{L}_{qr}.
\end{equation}

\paragraph{Geometric Inertia Matrix} Collecting the three geometric statistics $\mathcal{L}_{qq}$, $\mathcal{L}_{rr}$, and $\mathcal{L}_{qr}$ from \eqref{eq:Lqq_def}, \eqref{eq:Qpp_final}, and \eqref{eq:Lqr_identity}, we define the \textit{geometric inertia matrix}
\begin{equation}\label{eq:Lgeo_mat}
\bm{\mathcal{L}}_{\text{geo}}(\mathbf{P},\phi) \triangleq
\begin{bmatrix}
\mathcal{L}_{qq} & \mathcal{L}_{qr} \\
\mathcal{L}_{qr} & \mathcal{L}_{rr}
\end{bmatrix},
\end{equation}
where the entries are explicitly given by
\begin{subequations}
\begin{align}
\mathcal{L}_{qq} &= \sum_{m=1}^M (q_m - \bar{q})^2, \label{eq:Lqq}\\
\mathcal{L}_{rr} &= \sum_{m=1}^M (r_m - \bar{r})^2, \label{eq:Lrr}\\
\mathcal{L}_{qr} &= \sum_{m=1}^M (q_m - \bar{q})(r_m - \bar{r}). \label{eq:Lqr}
\end{align}
\end{subequations}
Note that $\bm{\mathcal{L}}_{\text{geo}}$ is a $2 \times 2$ symmetric positive semi-definite matrix. Its dependence on $\phi$ arises through the rotation \eqref{eq:rotation}: varying $\phi$ rotates the projection axes, redistributing the port spread between $\mathcal{L}_{qq}$ and $\mathcal{L}_{rr}$. Crucially, the \textit{trace} $\tr(\bm{\mathcal{L}}_{\text{geo}}) = \mathcal{L}_{qq} + \mathcal{L}_{rr} = \sum_{m}[(x_m - \bar{x})^2 + (y_m - \bar{y})^2]$ and the \textit{determinant} $\det(\bm{\mathcal{L}}_{\text{geo}}) = \mathcal{L}_{qq}\mathcal{L}_{rr} - \mathcal{L}_{qr}^2$ are both invariant to $\phi$, since unitary rotation preserves eigenvalues of the scatter matrix.

\paragraph{Assembly of the FIM}
Substituting the quadratic-form results $\mathcal{Q}_{\theta\theta}$, $\mathcal{Q}_{\phi\phi}$, and $\mathcal{Q}_{\theta\phi}$ into the general FIM expression \eqref{eq:fim_2d} with $\frac{2TP_s}{\sigma_n^2} = 2T\SNR$, we have
\begin{subequations}
\begin{align}
J_{\theta\theta}&= 2T \SNR 4\pi^2(\cos^2\theta) \mathcal{L}_{qq}\notag\\
&= 8\pi^2 T \SNR (\cos^2\theta)\mathcal{L}_{qq}, \label{eq:Jtt}\\
J_{\phi\phi}&= 2T \SNR 4\pi^2(\sin^2\theta) \mathcal{L}_{rr}\notag\\
&= 8\pi^2 T \SNR (\sin^2\theta)\mathcal{L}_{rr}, \label{eq:Jpp}\\
J_{\theta\phi}&= 2T \SNR 4\pi^2\cos\theta\sin\theta\, \mathcal{L}_{qr}\notag\\
&= 8\pi^2 T\SNR \cos\theta\sin\theta\;\mathcal{L}_{qr}. \label{eq:Jtp}
\end{align}
\end{subequations}
The FIM thus takes the compact matrix form
\begin{equation}\label{eq:FIM_matrix}
\mathbf{J} = 8\pi^2 T \SNR \times
\begin{bmatrix}
\cos^2\theta\,\mathcal{L}_{qq} & \cos\theta\sin\theta\,\mathcal{L}_{qr} \\
\cos\theta\sin\theta\,\mathcal{L}_{qr} & \sin^2\theta\,\mathcal{L}_{rr}
\end{bmatrix}.
\end{equation}
The structure of \eqref{eq:FIM_matrix} reveals a clean separation: the \textit{trigonometric prefactors} $\cos^2\theta$, $\sin^2\theta$, $\cos\theta\sin\theta$ capture the angular geometry of the wavefront, while the \textit{geometric inertia entries} $\mathcal{L}_{qq}$, $\mathcal{L}_{rr}$, $\mathcal{L}_{qr}$ capture the spatial distribution of ports. This factorization is the two-dimensional generalization of the scalar product $k^2\sin^2\theta\, \mathcal{L}_{\text{geo}}(\mathbf{p})$ in the linear-array case, where $k=2\pi/\lambda$ denotes the wavenumber \cite{Zhang2025}.

\paragraph{Closed-Form CRB for Elevation Angle} The CRB for the individual parameters can be obtained from the diagonal elements of the inverse FIM \cite[Chapter~3]{Kay1993}. For a $2 \times 2$ matrix, the inversion can be explicitly obtained by
\begin{equation}\label{eq:2x2_inv}
\mathbf{J}^{-1} = \frac{1}{\det(\mathbf{J})}
\begin{bmatrix}
J_{\phi\phi} & -J_{\theta\phi} \\
-J_{\theta\phi} & J_{\theta\theta}
\end{bmatrix},
\end{equation}
where the determinant is given by
\begin{align}\label{eq:det_J}
\det(\mathbf{J})
&= J_{\theta\theta}\, J_{\phi\phi} - J_{\theta\phi}^2 \nonumber \\
&= \big(8\pi^2 T \SNR\big)^2 \cos^2\theta\sin^2\theta \big(\mathcal{L}_{qq}\mathcal{L}_{rr} - \mathcal{L}_{qr}^2\big) \nonumber \\
&= \big(8\pi^2 T \SNR\big)^2 \cos^2\theta\sin^2\theta\, \det\!\big(\bm{\mathcal{L}}_{\text{geo}}\big).
\end{align}
The CRB for the elevation angle is $\mathrm{CRB}(\theta) = [\mathbf{J}^{-1}]_{11} = J_{\phi\phi} / \det(\mathbf{J})$. Substituting \eqref{eq:Jpp} and \eqref{eq:det_J} yields
\begin{align}\label{eq:crb_theta_derivation}
\mathrm{CRB}(\theta)&= \frac{8\pi^2 T \SNR (\sin^2\theta) \mathcal{L}_{rr}}{\big(8\pi^2 T \SNR\big)^2 \cos^2\theta\sin^2\theta\, \det(\bm{\mathcal{L}}_{\text{geo}})} \nonumber \\
&= \frac{\mathcal{L}_{rr}}{8\pi^2 T \SNR \cos^2\theta\, \det(\bm{\mathcal{L}}_{\text{geo}})}.
\end{align}
Writing $\det(\bm{\mathcal{L}}_{\text{geo}}) = \mathcal{L}_{qq}\mathcal{L}_{rr} - \mathcal{L}_{qr}^2$ explicitly, we arrive at the closed-form result
\begin{equation}\label{eq:crb_theta_final}
\mathrm{CRB}(\theta) = \frac{1}{8\pi^2 T \SNR \cos^2\theta} \times \frac{\mathcal{L}_{rr}}{\mathcal{L}_{qq}\mathcal{L}_{rr} - \mathcal{L}_{qr}^2}.
\end{equation}

Similarly, the CRB for the azimuth angle is $\mathrm{CRB}(\phi) = [\mathbf{J}^{-1}]_{22} = J_{\theta\theta} / \det(\mathbf{J})$ which can be found as
\begin{align}\label{eq:crb_phi_derivation}
\mathrm{CRB}(\phi)&= \frac{J_{\theta\theta}}{\det(\mathbf{J})}= \frac{8\pi^2 T \SNR (\cos^2\theta) \mathcal{L}_{qq}}{(8\pi^2 T  \SNR)^2 \cos^2\theta\sin^2\theta\, \det(\bm{\mathcal{L}}_{\text{geo}})} \nonumber \\
&= \frac{1}{8\pi^2 T \SNR \sin^2\theta} \times \frac{\mathcal{L}_{qq}}{\mathcal{L}_{qq}\mathcal{L}_{rr} - \mathcal{L}_{qr}^2}.
\end{align}
\end{IEEEproof}

\begin{remark}[Reduction to the Linear-Array Case]
{\em When the array degenerates to a linear topology along the $x$-axis (i.e., $y_m = 0$ for all $m$) and the azimuth is fixed at $\phi = 0$, the projected positions reduce to $q_m = x_m$ and $r_m = 0$. In this case, $\mathcal{L}_{rr} = 0$ and $\mathcal{L}_{qr} = 0$, so the FIM becomes singular in the $\phi$-direction (the linear array cannot resolve azimuth). The estimation problem reduces to a scalar one with $J_{\theta\theta} = 8\pi^2 T \SNR (\cos^2\theta) \mathcal{L}_{qq}$, and the single-parameter CRB is
\begin{equation}\label{eq:crb_1d_recovery}
\mathrm{CRB}(\theta)\big|_{\rm 1\text{-}D} = \frac{1}{8\pi^2 T \SNR (\cos^2\theta) \mathcal{L}_{qq}},
\end{equation}
with $\mathcal{L}_{qq} = \sum_{m}(x_m - \bar{x})^2 = \mathcal{L}_{\text{geo}}(\mathbf{p})$. Recalling that the one-dimensional steering vector uses the convention $e^{-j2\pi p_m \cos\theta}$ (i.e., the direction cosine is $\cos\theta$ rather than $\sin\theta\cos\phi$), the angular sensitivity factor $\cos^2\theta$ here corresponds to $\sin^2\theta$ in the linear-array formulation \cite{Zhang2025}, where $\theta$ is measured from the endfire direction. The geometric content $\mathcal{L}_{\text{geo}}$ is identical, confirming that \eqref{eq:crb_theta_final} is a consistent generalization.}
\end{remark}

\begin{remark}[Geometric Variables]\label{remark:det_interpretation}
\textit{The determinant $\det(\bm{\mathcal{L}}_{\text{geo}}) = \mathcal{L}_{qq}\mathcal{L}_{rr} - \mathcal{L}_{qr}^2$ is the \textbf{generalized geometric variance}, which is equal to the squared area of the port-position scatter ellipse. It is the product of the two eigenvalues of $\bm{\mathcal{L}}_{\text{geo}}$, denoted by $\ell_1 \ge \ell_2 \ge 0$, so $\det(\bm{\mathcal{L}}_{\text{geo}}) = \ell_1 \ell_2$. Maximizing $\det(\bm{\mathcal{L}}_{\text{geo}})$ simultaneously lowers both $\mathrm{CRB}(\theta)$ and $\mathrm{CRB}(\phi)$. When the cross-inertia $\mathcal{L}_{qr} = 0$ (which holds, e.g., for any placement symmetric about both the $q$- and $r$-axes), the FIM becomes diagonal and the two angles are decoupled. That is,
\begin{equation}
\mathcal{L}_{qr} = 0 \implies \mathrm{CRB}(\theta) = \frac{1}{8\pi^2 T \SNR (\cos^2\theta) \mathcal{L}_{qq}},
\end{equation}
recovering the planar analogy of the one-dimensional result. The decoupling condition $\mathcal{L}_{qr} = 0$ is equivalent to requiring that the principal axes of the port scatter ellipse align with the azimuth look direction and its perpendicular---a property automatically satisfied by rectangular grids aligned with $\phi$.}
\end{remark}

\begin{remark}[Geometric CRB Interpretation]\label{remark:crb_scaling}
\textit{From \eqref{eq:crb_theta_final}, the CRB scales inversely with the product $T\SNR$, confirming the classical behavior. More notably, the \textbf{geometry-dependent factor} $\mathcal{L}_{rr} / \det(\bm{\mathcal{L}}_{\text{geo}})$ reveals that it is not the total geometric spread alone but its \textbf{effective component along the elevation-sensitive axis} that governs precision. Specifically, $\mathrm{CRB}(\theta)$ depends on $\mathcal{L}_{qq}$ only through $\det(\bm{\mathcal{L}}_{\text{geo}})$, while $\mathcal{L}_{rr}$ appears in the numerator. A placement that is elongated along the $q$-direction (large $\mathcal{L}_{qq}$, and small $\mathcal{L}_{rr}$) improves elevation estimation but degrades azimuth estimation, and vice versa. The isotropic optimum ($\ell_1 = \ell_2$, i.e., circular scatter) minimizes $\max\{\mathrm{CRB}(\theta), \mathrm{CRB}(\phi)\}$ over all rotations.}
\end{remark}

\begin{remark}[Precision-Ambiguity Tradeoff]\label{remark:tradeoff_2d}
\textit{As in the linear case, a fundamental trade-off between \textbf{estimation precision} and \textbf{ambiguity suppression} persists. From \eqref{eq:crb_theta_final}, concentrating ports near the four corners of the aperture maximizes $\det(\bm{\mathcal{L}}_{\text{geo}})$ since positions far from the centroid contribute the largest squared deviations and hence minimizes the CRB. However, such corner-heavy configurations produce a highly sparse array with large gaps, which amplifies \textbf{false spectral peaks} (grating lobes) in the two-dimensional spatial spectrum and causes severe bearing misdirection. Conversely, a uniform rectangular grid distributes ports evenly, filling the virtual co-array and suppressing false peaks, but at the expense of a \textbf{smaller effective geometric spread}, yielding a higher CRB. This precision-ambiguity trade-off motivates the joint optimization to be investigated in Section~\ref{sec:algorithm_design}.}
\end{remark}

\begin{remark}[Mutual Coupling Effect]\label{remark:mutual_coupling}
\textit{The CRB derived above assumes an ideal array manifold. The estimation precision is governed entirely by the geometric factor $\bm{\mathcal{L}}_{\text{geo}}(\mathbf{P}, \phi)$. Under strong electromagnetic mutual coupling, the effective steering vector deviates from \eqref{eq:steer_2d}, and the FIM becomes hardware-dependent, and a coupled antenna/network model is required \cite{mc1,mc2,mc3}. The present result should therefore be viewed as a geometry-level benchmark that isolates the spatial-design contribution from electromagnetic effects.}
\end{remark}

\begin{algorithm}[t!]
\small
\caption{Regularized Greedy Port Selection for CRB Minimization}\label{alg:greedy}
\KwIn{$M$, aperture $W_x \times W_y$, grid spacing $\Delta$, min spacing $d_{\min}$, diversity weight $\beta$.}
\KwOut{Optimized port positions $\mathbf{P}^*$.}

\tcc{Phase~I: Candidate Generation}
Generate grid $\mathcal{G}$ via \eqref{eq:grid}\;
$\mathcal{C} \leftarrow \{\bm{g} \in \mathcal{G} : \|\bm{g} - \bm{p}_c\| \ge d_{\min},\; \forall\text{ corner } \bm{p}_c\} \setminus \mathcal{S}_0$\;

\tcc{Phase~II: Regularized Greedy Selection}
$\mathcal{S} \leftarrow \mathcal{S}_0 = \{\text{four corners}\}$\;
Initialize running sums $S_{xx}, S_{yy}, S_{xy}, S_x, S_y$ from corner positions\;
\For{$k = 1$ \KwTo $M - 4$}{
$\mathcal{C}_k \leftarrow \{\bm{g} \in \mathcal{C} \setminus \mathcal{S} : \|\bm{g} - \bm{s}\| \ge d_{\min},\; \forall\, \bm{s} \in \mathcal{S}\}$\;
\ForEach{$\bm{g} \in \mathcal{C}_k$}{
Compute $\det(\bm{\mathcal{L}}_{\text{geo}}(\mathcal{S} \cup \{\bm{g}\}))$ via incremental update of \eqref{eq:det_from_sums}\;
Compute $d_{\min}^2(\bm{g}) = \min_{\bm{s} \in \mathcal{S}} \|\bm{g} - \bm{s}\|^2$\;
$\mathrm{score}(\bm{g}) \leftarrow \det(\bm{\mathcal{L}}_{\text{geo}}(\mathcal{S} \cup \{\bm{g}\})) + \beta \cdot d_{\min}^2(\bm{g})$\;
}
$\bm{g}^* \leftarrow \arg\max_{\bm{g} \in \mathcal{C}_k} \mathrm{score}(\bm{g})$\;
$\mathcal{S} \leftarrow \mathcal{S} \cup \{\bm{g}^*\}$\;
Update $S_{xx}, S_{yy}, S_{xy}, S_x, S_y$\;
}
$\mathbf{P}^* \leftarrow \mathcal{S}$\;
\end{algorithm}

\section{Array Placement Algorithm Design}\label{sec:algorithm_design}
Remark~\ref{remark:tradeoff_2d} established the fundamental trade-off between estimation precision and spatial ambiguity suppression. Here, we design a computationally efficient placement algorithm that navigates this trade-off by maximizing the CRB-relevant geometric determinant while maintaining spatial diversity across the aperture interior. The proposed algorithm provides a tool to further validate the discoveries in this work.

\subsection{Discrete Candidate Formulation}
Once a minimum inter-port distance $d_{\min}$ is prescribed, the continuous aperture $[0, W_x] \times [0, W_y]$ is discretized into a regular grid with spacing $\Delta \le d_{\min}$, yielding a candidate set
\begin{equation}\label{eq:grid}
\begin{aligned}
\mathcal{G} = \big\{(i\Delta,\, j\Delta) : 
&\ i = 0,1,\dots,\lfloor W_x/\Delta\rfloor,\\
&\ j = 0,1,\dots,\lfloor W_y/\Delta\rfloor\big\}.
\end{aligned}
\end{equation}
A candidate $\bm{g} \in \mathcal{G}$ is \textit{admissible} only if it respects the minimum distance, i.e., $\|\bm{g} - \bm{p}_c\| \ge d_{\min}$ for each corner $\bm{p}_c$ in \eqref{eq:corner_constraint}. Let $\mathcal{C} \subseteq \mathcal{G}$ denote the filtered admissible set with $|\mathcal{C}| = N_c$ denoting the number of admissible candidates.

\subsection{Limitation of Pure Determinant-Optimal Selection}
From \eqref{eq:crb_theta_final}, minimizing $\mathrm{CRB}(\theta)$ is equivalent to maximizing $\det(\bm{\mathcal{L}}_{\text{geo}})$. However, this \textit{determinant-optimal} (D-optimal) criterion measures only the aggregate scatter of port positions around their centroid. Since boundary positions always provide the largest squared deviations, a pure D-optimal greedy algorithm will concentrate all ports on the aperture periphery, leaving the {\em interior void} (i.e., no ports will be placed or activated in the center area). While this minimizes the CRB in isolation, the resulting placement produces a highly non-uniform virtual co-array with large spatial gaps, which can increase sidelobe levels and spectral ambiguity \cite{Moffet1968}. The same precision-ambiguity trade-off identified in Remark~\ref{remark:tradeoff_2d} manifests here as a concrete algorithmic artifact.

\subsection{Regularized Greedy Selection}
To balance CRB minimization with spatial diversity, we augment the greedy criterion with a \textit{minimum-separation reward} (regularized component) that encourages port dispersion across the aperture interior. At each selection stage $k$, the candidate is chosen as
\begin{equation}\label{eq:reg_greedy}
\bm{g}^{(k)} = \mathop{\arg\max}\limits_{\bm{g} \in \mathcal{C}_k}\;\underbrace{\det\!\big(\bm{\mathcal{L}}_{\text{geo}}(\mathcal{S} \cup \{\bm{g}\})\big)}_{\text{precision (D-optimality)}}
+ \;\beta\;\underbrace{\min_{\bm{s} \in \mathcal{S}} \|\bm{g} - \bm{s}\|^2}_{\text{spatial diversity}},
\end{equation}
where $\mathcal{C}_k = \{\bm{g} \in \mathcal{C} \setminus \mathcal{S} : \|\bm{g} - \bm{s}\| \ge d_{\min}, \forall \bm{s} \in \mathcal{S}\}$ is the feasible subset, and $\beta \ge 0$ is the \textit{diversity weight} that governs the precision-diversity trade-off. The two limiting cases are:
\begin{itemize}
\item $\beta = 0$: this corresponds to pure D-optimal selection and all ports tend to be on the boundary;
\item $\beta \to \infty$: this is the maximum-dispersion packing where ports are spread as far apart as possible.
\end{itemize}
A moderate $\beta$ in \eqref{eq:reg_greedy} steers the algorithm to place a subset of ports in the aperture interior, populating the co-array more uniformly without excessively sacrificing geometric spread.

\subsection{Efficient $\mathcal{O}(1)$ Evaluation}
Since $\bm{\mathcal{L}}_{\text{geo}} \in \R^{2 \times 2}$, its determinant is evaluated from the running sums $S_{xx} = \sum x_m^2$, $S_{yy} = \sum y_m^2$, $S_{xy} = \sum x_m y_m$, $S_x = \sum x_m$, $S_y = \sum y_m$ using
\begin{equation}\label{eq:det_from_sums}
\det(\bm{\mathcal{L}}_{\text{geo}}) = \Big(S_{xx} - \frac{S_x^2}{M_k}\Big)\Big(S_{yy} - \frac{S_y^2}{M_k}\Big) - \Big(S_{xy} - \frac{S_x S_y}{M_k}\Big)^2,
\end{equation}
where $M_k = |\mathcal{S}|$ is the current port count. Adding a candidate $\bm{g} = (g_x, g_y)^\trans$ requires updating each running sum by a single addition. The minimum-distance term $\min_{\bm{s} \in \mathcal{S}} \|\bm{g} - \bm{s}\|^2$ costs $\mathcal{O}(M_k)$ per candidate. The total complexity is therefore $\mathcal{O}((M-4)  N_c  M)$, which remains negligible.

\begin{table}[t]
\centering
\caption{Simulation Parameters}\label{tab:sim_params}
\renewcommand{\arraystretch}{1.15}
\begin{tabular}{l l l}
\hline
\textbf{Parameter} & \textbf{Symbol} & \textbf{Value} \\
			\hline
			Wavelength & $\lambda$ & 1 (normalized) \\
			Default aperture & $W_x \times W_y$ & $2\lambda \times 2\lambda$ \\
			Default port count & $M$ & 25 \\
			Minimum spacing & $d_{\min}$ & $0.2\lambda$ \\
			Candidate grid spacing & $\Delta$ & $d_{\min}/2 = 0.1\lambda$ \\
			Target elevation & $\theta_0$ & $45^\circ$ \\
			Target azimuth & $\phi_0$ & $30^\circ$ \\
			Snapshot count & $T$ & 100 \\
			Operating SNR & --- & 10~dB \\
			SNR sweep range & --- & $-10$ to $30$~dB (step 2~dB) \\
			Diversity weights (Fig.~\ref{fig:placement_beam}) & $\beta_0$ & $\{0,\, 5,\, 10,\, 100\}$ \\
			Representative $\beta_0$ (Fig.~\ref{fig:task2}) & $\beta_0$ & 0.8 \\
			$\beta_0$ sweep (Fig.~\ref{fig:tradeoff}) & $\beta_0$ & $[0, 5]$, 50 points \\
			MC trials (PDF validation) & --- & $10^5$ \\
			MC trials (random baseline) & --- & 500 \\
			MC trials (multi-$(W,M)$) & --- & 200 \\
			Beam pattern grid & $N_{uv}$ & $301 \times 301$ \\
			dB floor (beam plot) & --- & $-30$~dB \\
			\hline
			Multi-$(W,M)$ configurations & $(W, M)$ & $(1\lambda, 5)$, $(2\lambda, 25)$, \\
			& & $(4\lambda, 55)$, $(6\lambda, 85)$ \\
			\hline
\end{tabular}
\end{table}

\begin{figure}[t!]
\centering
\includegraphics[width=\columnwidth]{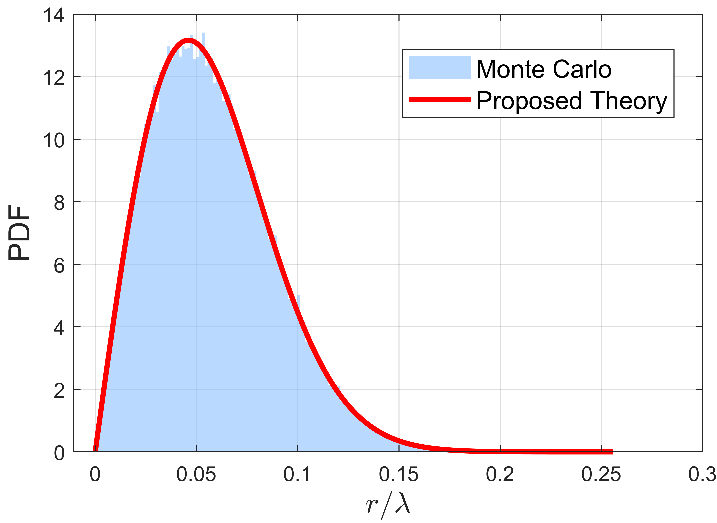}
\caption{Empirical and theoretical PDF of the minimum inter-port distance $\Delta_{\min}$ for $M = 25$ ports uniformly distributed over a $2\lambda \times 2\lambda$ aperture ($\mathcal{A} = 4\lambda^2$).}\label{fig:pdf}
\end{figure}

\begin{figure}[t!]
\centering
\includegraphics[width=\columnwidth]{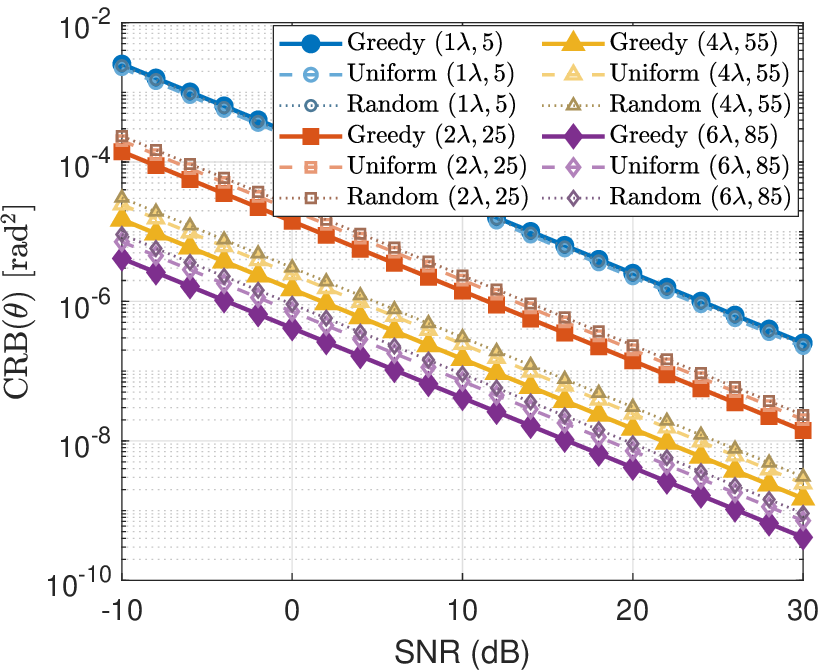}
\caption{CRB vs.\ SNR and geometric determinant for multiple $(W, M)$ configurations with $\beta_0 = 0.8$.}\label{fig:task2}
\end{figure}

\begin{figure}[t!]
\centering
\includegraphics[width=1\columnwidth]{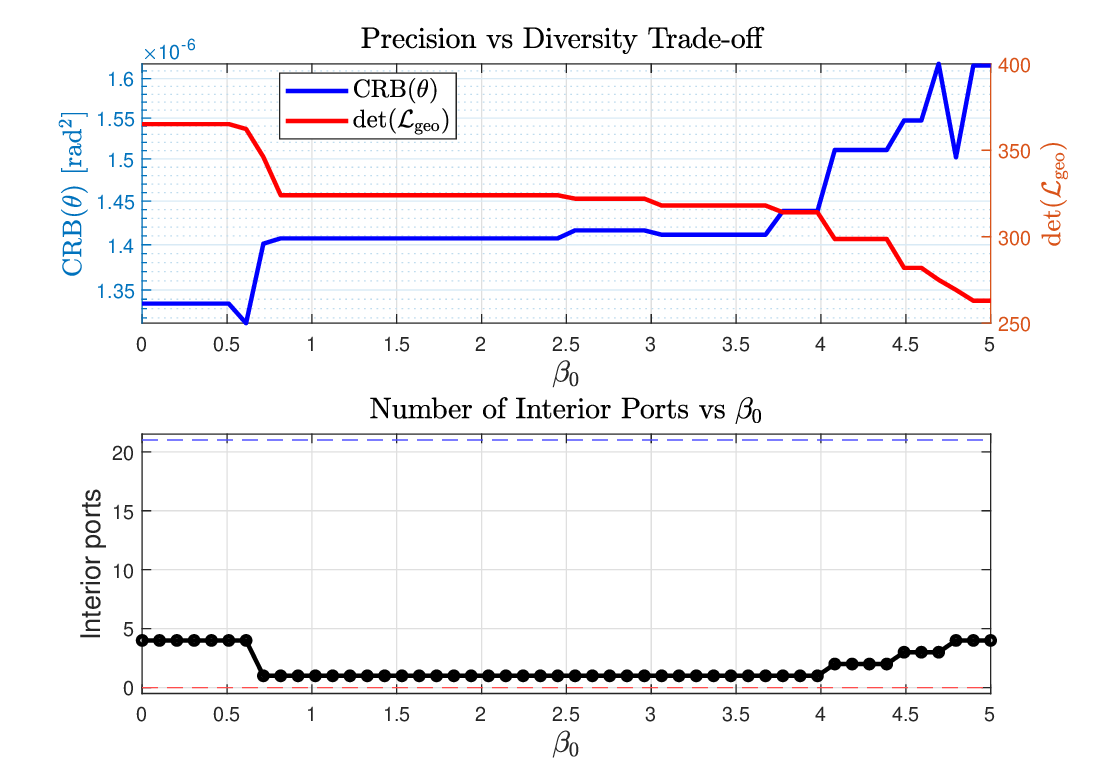}
\caption{Trade-off curves vs.\ $\beta_0$ for the default configuration: (top) CRB and $\det(\bm{\mathcal{L}}_{\mathrm{geo}})$; (bottom) number of interior ports.}\label{fig:tradeoff}
\end{figure}

\subsection{Selection of $\beta$}
The diversity weight $\beta$ controls the fraction of interior ports. To achieve a balance across different setups, we normalize both terms and set
\begin{equation}\label{eq:beta_selection}
\beta = \beta_0 \frac{\det(\bm{\mathcal{L}}_{\text{geo}}(\mathcal{S}_0))}{\mathcal{A}},
\end{equation}
where $\beta_0$ is a dimensionless tuning parameter and $\mathcal{S}_0$ is the initial four-corner set. The rationale is as follows: at the first greedy step, the precision term in~\eqref{eq:reg_greedy} is on the order of $\det(\bm{\mathcal{L}}_{\text{geo}}(\mathcal{S}_0))$, while the diversity term $\min_m \|\mathbf{g}_j - \mathbf{p}_m\|^2$ is on the order of the aperture area~$\mathcal{A}$ (since the squared distances are bounded by $W_x^2 + W_y^2$). Dividing the former by the latter yields a scaling factor that makes the two terms comparable in magnitude when $\beta_0 = 1$, so that $\beta_0$ directly controls the relative emphasis between precision and diversity regardless of the aperture size or port count (i.e., $\beta_0 = 1$ balances the two objectives equally and $\beta_0 > 1$ biases toward spatial diversity while $\beta_0 < 1$ biases toward estimation precision). The complete procedure is summarized in Algorithm~\ref{alg:greedy}.

\subsection{Computational Complexity}
Phase~I runs in $\mathcal{O}(N_c)$. Phase~II iterates $(M-4)$ rounds, each scanning at most $N_c$ candidates with $\mathcal{O}(M)$ operations per candidate, yielding a total of $\mathcal{O}((M-4) \cdot N_c \cdot M)$.

\begin{remark}[Regularized Term $\beta$]\label{remark:beta_effect}
\textit{The diversity weight $\beta$ provides control over the precision-diversity trade-off. As will be shown in Section~\ref{sec:numerical}, increasing $\beta$ from $0$ to moderate values initially causes only a minor CRB increase (as the boundary ports already dominate the geometric spread) while substantially improving the spatial uniformity of the placement. This graceful degradation allows the system designer to tune $\beta$ according to the relative importance of CRB versus ambiguity robustness in the target application scenario.}
\end{remark}

\begin{figure*}[t!]
\centering
\includegraphics[width=\linewidth]{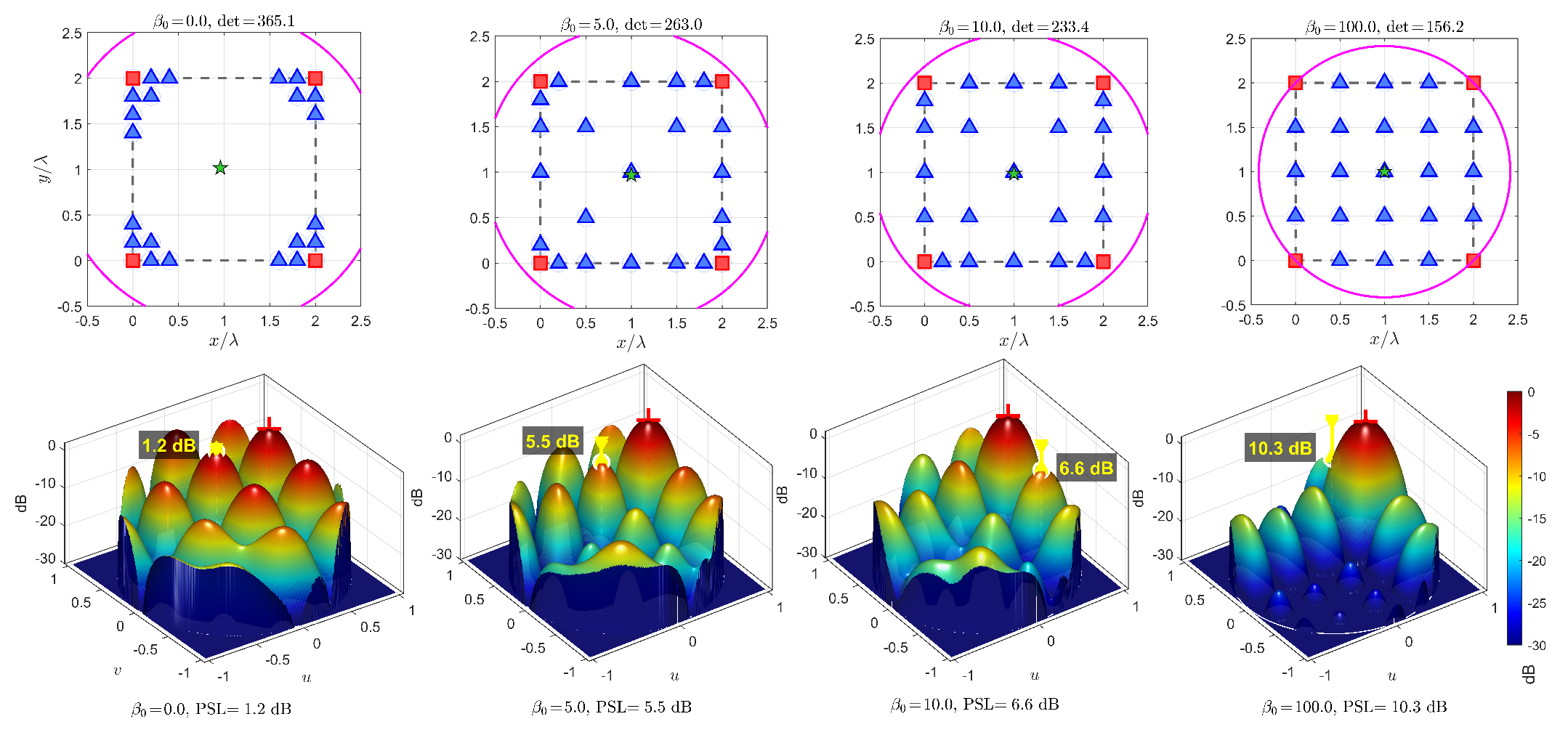}
\caption{Port placement (top row) and beam pattern in the $(u, v)$ plane (bottom row) for $\beta_0 \in \{0,\, 5,\, 10,\, 100\}$ under the default configuration ($W = 2\lambda$, $M = 25$). The yellow annotation indicates the PSL relative to the main beam.}\label{fig:placement_beam}
\end{figure*}

\section{Numerical Results}\label{sec:numerical}
In this section, we evaluate the analytical results derived and the proposed regularized greedy algorithm for finite-aperture planar FAA port selection. All port positions are normalized by the wavelength ($\lambda = 1$). Unless otherwise stated, the default configuration is a square aperture of size $W_x \times W_y = 2\lambda \times 2\lambda$ with $M = 25$ ports, minimum inter-port spacing $d_{\min} = 0.2\lambda$, target direction $(\theta_0, \phi_0) = (45^\circ, 30^\circ)$, snapshot count $T = 100$, and that the operating SNR is set to $10$~dB.

\subsection{Simulation Setup}
Regarding the candidate grid, port positions are selected from a uniform rectangular grid with spacing $\Delta = d_{\min}/2 = 0.1\lambda$ over the aperture $[0, W_x] \times [0, W_y]$. Grid points that violate the minimum distance constraint $d_{\min}$ with respect to any of the four corner ports are excluded. The four corner ports $\{(0,0),\,(W_x,0),\,(0,W_y),\,(W_x,W_y)\}$ are always included to maximize the aperture utilization, and the remaining $M - 4$ ports are selected by the algorithm.

Two placement strategies are provided as benchmarks:
\begin{itemize}
\item {\bf Uniform Grid}---Ports are placed on a regular $M_x \times M_y$ rectangular grid within $[0, W_x] \times [0, W_y]$, in which $(M_x, M_y)$ is chosen such that $M_x M_y \ge M$ with minimum excess. If $M_x M_y > M$, corner ports are retained and the most central excess ports are removed.
\item {\bf Random Placement}---Each realization draws $M$ ports uniformly at random within the aperture, ensuring that the $d_{\min}$ spacing constraints are not violated.
\end{itemize}

For the default configuration, four values of the diversity weight $\beta_0 \in \{0,\, 5,\, 10,\, 100\}$ are compared to visualize the precision-ambiguity trade-off (Remark~5). The beam pattern is computed in the direction-cosine $(u, v)$ plane via the steered array factor, which is calculated as
\begin{equation}
B(u,v) = \left| \frac{1}{M} \sum_{m=1}^{M} e^{\,j2\pi\bigl[x_m(u - u_0) + y_m(v - v_0)\bigr]} \right|^2,
\end{equation}
where $u_0 = \sin\theta_0 \cos\phi_0$ and $v_0 = \sin\theta_0 \sin\phi_0$. The peak sidelobe level (PSL) is measured by comparing the main lobe with the maximum of the remaining pattern.

The following metrics are reported:
\begin{itemize}
\item $\det(\bm{\mathcal{L}}_{\mathrm{geo}})$: geometric determinant of the scatter matrix, governing estimation precision (Proposition~\ref{prop:CRB}).
\item $\mathrm{CRB}(\theta)$: CRB on the elevation angle, computed as $\mathrm{CRB}(\theta) = L_{rr} / (8\pi^2 T\mathrm{SNR} (\cos^2\!\theta_0)  \det(\bm{\mathcal{L}}_{\mathrm{geo}}))$.
\item PSL (dB): peak sidelobe level relative to the main beam, quantifying spatial ambiguity.
\item Number of interior ports: counting ports not on the aperture boundary (with margin $= d_{\min}/2$), indicating the degree of spatial diversity.
\end{itemize}

The parameter settings are detailed in Table~\ref{tab:sim_params}.

\subsection{Results and Discussion}
Fig.~\ref{fig:pdf} validates the Rayleigh approximation of Proposition~\ref{prop:2d_min_dist} for the planar case. The Monte Carlo histogram ($10^5$ realizations) closely matches the theoretical PDF $f(r) = 2\alpha r\, e^{-\alpha r^2}$ with $\alpha = M(M-1)\pi / (2\mathcal{A})$, confirming that the expected minimum spacing $\mathbb{E}[\Delta_{\min}] = \frac{1}{2}\sqrt{2\mathcal{A}/[M(M-1)]}$ provides a statistically grounded reference for choosing $d_{\min}$ when the physical port size is not explicitly modeled.

Fig.~\ref{fig:task2} compares the CRB performance across four $(W, M)$ configurations. First, increasing the aperture $W$ while scaling the port count proportionally (to maintain comparable spatial density) consistently reduces CRB across all SNR levels, confirming that the geometric determinant $\det(\bm{\mathcal{L}}_{\mathrm{geo}})$ scales favorably with aperture size. Additionally, the proposed greedy algorithm consistently outperforms both the uniform grid and random placements in terms of $\det(\bm{\mathcal{L}}_{\mathrm{geo}})$, with the gap widening for larger $(W, M)$ pairs. 

Fig.~\ref{fig:tradeoff} traces the trade-off continuously over $\beta_0 \in [0, 5]$. The CRB increases monotonically (precision degrades) while $\det(\bm{\mathcal{L}}_{\mathrm{geo}})$ decreases, confirming the fundamental tension between estimation accuracy and spatial diversity. The number of interior ports exhibits a step-like transition: for small $\beta_0$, all free ports remain on the boundary, but beyond a critical threshold they rapidly migrate inward. This sharp transition suggests that in practice, a moderate $\beta_0$ suffices to achieve a favorable balance, avoiding the extreme of either all-boundary placement (high precision, high ambiguity) or fully dispersed placement (low ambiguity, degraded precision).

Finally, Fig.~\ref{fig:placement_beam} visualizes the precision-ambiguity trade-off described in Remark~5. As we can see, for $\beta_0 = 0$ (pure D-optimal), all free ports cluster near the aperture boundary, maximizing $\det(\bm{\mathcal{L}}_{\mathrm{geo}})$ but producing severe grating lobes (low PSL). As $\beta_0$ increases, the diversity regularizer progressively pushes ports toward the interior, yielding a more uniform spatial distribution. The beam patterns in the bottom row clearly illustrate the sidelobe suppression: the PSL improves monotonically with $\beta_0$, while $\det(\bm{\mathcal{L}}_{\mathrm{geo}})$ decreases (wider main beam). The scatter ellipses in the top row shrink toward a circle as the placement becomes more isotropic, consistent with the convergence of $\bm{\mathcal{L}}_{\mathrm{geo}}$ toward a scaled identity.

\section{Conclusion}
This work established a systematic analytical framework for finite-aperture planar FAAs. We proved that the minimum inter-port distance under random planar placement follows a Rayleigh law scaling as $\mathcal{O}(M^{-1})$, quadratically more favorable than the linear-array counterpart. We derived universal closed-form CRBs governed by a $2\times2$ geometric inertia matrix $\bm{\mathcal{L}}_{\mathrm{geo}}$ with azimuth-invariant trace and determinant, and identified a fundamental precision-ambiguity trade-off intrinsic to finite-aperture placement. A regularized greedy algorithm was proposed to navigate this trade-off, consistently outperforming uniform-grid and random baselines. These results confirm that geometry diversity in planar fluid antennas is fully predictable and analytically tractable, offering a deterministic performance dimension complementary to fading diversity.

Future work will focus on jointly optimizing port positions and beamforming weights with explicit peak-sidelobe-level constraints, converting the current indirect trade-off into a directly controllable multi-objective design.


\end{document}